\documentclass[submission]{eptcs}
\providecommand{\event}{WORDS 2011} 
\usepackage{amssymb}
\usepackage{amsthm}
\usepackage{amsmath}

    \newtheorem{theorem}{Theorem}[section]
    \newtheorem{lemma}[theorem]{Lemma}
    \newtheorem{corollary}[theorem]{Corollary}
    \theoremstyle{definition}
    \newtheorem{example}[theorem]{Example}

    \newcommand{\eps}{\varepsilon}
    \newcommand{\dd}{\cdots}
    \newcommand{\Z}{\mathbb Z}
    \newcommand{\Q}{\mathbb Q}
    \newcommand{\No}{\mathbb N_0}
    \newcommand{\Ni}{\mathbb N_1}
    \newcommand{\set}[2]{ \left\{ #1 \mid #2 \right\} }
    \newcommand{\px}[1]{ P_{#1} }
    \newcommand{\qx}[3]{ Q_{#1,#2,#3} }
    \newcommand{\sx}[2]{ S_{#1,#2} }
    \newcommand{\len}[1]{\mathrm{len}_{#1}}
    \newcommand{\lhp}{\mathcal M}

\begin{document}

\title{Systems of Word Equations and Polynomials: \\ a New Approach
    \thanks{Supported by the Academy of Finland under grant 121419}
}
\author{Aleksi Saarela \institute{
    Turku Centre for Computer Science TUCS and Department of Mathematics \\
    University of Turku, FI-20014 Turku, FINLAND
    \email{amsaar@utu.fi}
} }
\def\titlerunning{Systems of Word Equations and Polynomials: a New Approach}
\def\authorrunning{Aleksi Saarela}

\maketitle

\begin{abstract}
We develop new polynomial methods for studying systems of word
equations. We use them to improve some earlier results and to
analyze how sizes of systems of word equations satisfying certain
independence properties depend on the lengths of the equations.
These methods give the first nontrivial upper bounds for the sizes
of the systems.
\end{abstract}

\section{Introduction}

Word equations are a fundamental part of combinatorics on words, see
e.g. \cite{Lo83} or \cite{ChKa97} for a general reference on these
subjects. One of the basic results in the theory of word equations
is that a nontrivial equation causes a defect effect. In other
words, if $n$ words satisfy a nontrivial relation, then they can be
represented as products of $n-1$ words. Not much is known about the
additional restrictions caused by several independent relations
\cite{HaKa04}.

In fact, even the following simple question, formulated already in
\cite{CuKa83}, is still unanswered: how large can an independent
system of word equations on three unknowns be? The largest known
examples consist of three equations. The only known upper bound
comes from the Ehrenfeucht Compactness Property, proved in
\cite{AlLa85} and independently in \cite{Gu86}: an independent
system cannot be infinite. This question can be obviously asked also
in the case of $n > 3$ unknowns. Then there are independent systems
of size $\Theta(n^4)$ \cite{KaPl96}. Some results concerning
independent systems on three unknowns can be found in \cite{HaNo03},
\cite{CzKa07} and \cite{CzPl09}, but the open problem seems to be
very difficult to approach with current techniques.

There are many variations of the above question: we may study it in
the free semigroup, i.e. require that $h(x) \ne \eps$ for every
solution $h$ and unknown $x$, or examine only the systems having a
solution of rank $n-1$, or study chains of solution sets instead of
independent systems. See e.g. \cite{HaKaPl02}, \cite{HaKa04},
\cite{Cz08} and \cite{KaSa11}.

In this article we will try to use polynomials to study some
questions related to systems of word equations. Algebraic techniques
have been used before, most notably in the proof of Ehrenfeucht's
conjecture, which is based on Hilbert's Basis Theorem. However, the
way in which we use polynomials is quite different and allows us to
apply linear algebra to the problems.

One of the main contributions of this article is the development of
new methods for attacking problems on word equations. This is done
in Sections \ref{sect:fixedlength} and \ref{sect:solsets}. Other
contributions include simplified proofs and generalizations for old
results in Sect. \ref{sect:appl} and in the end of Sect.
\ref{sect:solsets}, and studying maximal sizes of independent
systems of equations in Sect. \ref{sect:indsyst}. Thus the
connection between word equations and linear algebra is not only
theoretically interesting, but is also shown to be very useful at
establishing simple-looking results that have been previously
unknown, or that have had only very complicated proofs. In addition
to the results of the paper, we believe that the techniques may be
useful in further analysis of word equations.

Now we give a brief overview of the paper. First, in Sect.
\ref{sect:basic} we define a way to transform words into polynomials
and prove some basic results using these polynomials.

In Sect. \ref{sect:fixedlength} we prove that if the lengths of the
unknowns are fixed, then there is a connection between the ranks of
solutions of a system of equations and the rank of a certain
polynomial matrix. This theorem is very important for all the later
results.

Section \ref{sect:appl} contains small generalizations of two
earlier results. These are nice examples of the methods developed in
Sect. \ref{sect:fixedlength} and have independent interest, but they
are not important for the later sections.

In Sect. \ref{sect:solsets} we analyze the results of Sect.
\ref{sect:fixedlength}, when the lengths of the unknowns are not
fixed. For every solution these lengths form an $n$-dimensional
vector, called the \emph{length type} of the solution. We prove that
the length types of all solutions of rank $n-1$ of a pair of
equations are covered by a finite union of $(n-1)$-dimensional
subspaces, if the equations are not equivalent on solutions of rank
$n-1$. This means that the solution sets of pairs of equations are
in some sense more structured than the solution sets of single
equations. This theorem is the key to proving the remaining results.
We conclude Sect. \ref{sect:solsets} by proving a theorem about
unbalanced equations. This gives a considerably simpler reproof and
a generalization of a result in \cite{HaNo03}

Finally, in Sect. \ref{sect:indsyst} we return to the question about
sizes of independent systems. There is a trivial bound for the size
of a system depending on the length of the longest equation, because
there are only exponentially many equations of a fixed length. We
prove that if the system is independent even when considering only
solutions of rank $n-1$, then there is an upper bound for the size
of the system depending quadratically on the length of the shortest
equation. Even though it does not give a fixed bound even in the
case of three unknowns, it is a first result of its type -- hence
opening, we hope, a new avenue for future research.

\section{Basic Theorems} \label{sect:basic}

Let $|w|$ be the length of a word $w$ and $|w|_a$ be the number of
occurrences of a letter $a$ in $w$. We use the notation $u \leq v$,
if $u$ is a prefix of $v$. We denote the set of nonnegative integers
by $\No$ and the set of positive integers by $\Ni$. The empty word
is denoted by $\eps$.

In this section we give proofs for some well-known results. These
serve as examples of the polynomial methods used. Even though the
standard proofs of these are simple, we hope that the proofs given
here illustrate how properties of words can be formulated and proved
in terms of polynomials.

Let $\Sigma \subset \Ni$ be an alphabet of numbers. For a word $w =
a_{0} \dots a_{n-1} \in \Sigma^n$ we define a polynomial
\begin{equation*}
    \px{w} = a_{0} + a_1 X^{1} + \dots + a_{n-1} X^{n-1}.
\end{equation*}
Now $w \mapsto \px{w}$ is an injective mapping from words to
polynomials (here we need the assumption $0 \notin \Sigma$). If
$w_1, \dots, w_m \in \Sigma^*$, then
\begin{equation} \label{eq:prodp}
    \px{w_1 \dots w_m}
    = \px{w_1} + \px{w_2} X^{|w_1|}
        + \dots + \px{w_m} X^{|w_1 \dots w_{m-1}|} .
\end{equation}
If $w \in \Sigma^+$ and $k \in \No$, then
\begin{equation*}
    \px{w^k} = \px{w} \frac{X^{k|w|} - 1}{X^{|w|} - 1}
\end{equation*}

The polynomial $\px{w}$ can be viewed as a characteristic polynomial
of the word $w$. We could also replace $X$ with a suitable number
$b$ and get a number whose reverse $b$-ary representation is $w$. Or
we could let the coefficients of $\px{w}$ be from some other
commutative ring than $\Z$. Similar ideas have been used to analyze
words in many places, see e.g. \cite{Ku97}, \cite{Sa85} and
\cite{HoKo09}.

\begin{example}
If $w = 1212$, then
\begin{math}
    \px{w} = 1 + 2X + X^2 + 2X^3 .
\end{math}
\end{example}

A word $w$ is \emph{primitive}, if it is not of the form $u^k$ for
any $k > 1$. If $w = u^k$ and $u$ is primitive, then $u$ is a
\emph{primitive root} of $w$.

\begin{lemma} \label{lem:primdiv}
If $w$ is primitive, then $\px{w}$ is not divisible by any
polynomial of the form
\begin{math}
    (X^{|w|} - 1) / (X^{n} - 1),
\end{math}
where $n < |w|$ is a divisor of $|w|$.
\end{lemma}
\begin{proof}
If $\px{w}$ is divisible by $(X^{|w|} - 1) / (X^n - 1)$, then there
are numbers $a_0, \dots, a_{n-1}$ such that
\begin{equation*}
    \px{w}
    = (a_{0} + a_1 X^{1} + \dots + a_{n-1} X^{n-1})
        \frac{X^{|w|} - 1}{X^n - 1}
    = (a_{0} + a_1 X^{1} + \dots + a_{n-1} X^{n-1})
        (1 + X^n + \dots + X^{|w|-n}) ,
\end{equation*}
so $w = (a_{0} \dots a_{n-1})^{|w|/n}$.
\end{proof}

The next two theorems are among the most basic and well-known
results in combinatorics on words (except for item \eqref{item:r} of
Theorem \ref{thm:commutation}).

\begin{theorem}
Every nonempty word has a unique primitive root.
\end{theorem}
\begin{proof}
Let $u^m = v^n$, where $u$ and $v$ are primitive. We need to show
that $u = v$. We have
\begin{equation*}
    \px{u} \frac{X^{m|u|}-1}{X^{|u|}-1} = \px{u^m}
    = \px{v^n} = \px{v} \frac{X^{n|v|}-1}{X^{|v|}-1} .
\end{equation*}
Because $m|u| = n|v|$, we get
\begin{math}
    \px{u} (X^{|v|}-1) = \px{v} (X^{|u|}-1) .
\end{math}
If $d = \gcd(|u|, |v|)$, then $\gcd(X^{|u|}-1, X^{|v|}-1) = X^d-1$.
Thus $\px{u}$ must be divisible by $(X^{|u|}-1) / (X^d-1)$ and
$\px{v}$ must be divisible by $(X^{|v|}-1) / (X^d-1)$. By Lemma
\ref{lem:primdiv}, both $u$ and $v$ can be primitive only if $|u| =
d = |v|$.
\end{proof}

The primitive root of a word $w \in \Sigma^+$ is denoted by
$\rho(w)$.

\begin{theorem} \label{thm:commutation}
For $u,v \in \Sigma^+$, the following are equivalent:
\begin{enumerate}
\item $\rho(u) = \rho(v)$, \label{item:rho}
\item if $U,V \in \{u,v\}^*$ and $|U| = |V|$, then $U = V$, \label{item:all}
\item $u$ and $v$ satisfy a nontrivial relation, \label{item:exists}
\item $\px{u} / (X^{|u|} - 1) = \px{v} / (X^{|v|} - 1)$. \label{item:r}
\end{enumerate}
\end{theorem}
\begin{proof}
(\ref{item:rho}) $\Rightarrow$ (\ref{item:all}):
\begin{math}
    U = \rho(u)^{|U|/|\rho(u)|}
    = \rho(u)^{|V|/|\rho(u)|} = V .
\end{math}

(\ref{item:all}) $\Rightarrow$ (\ref{item:exists}): Clear.

(\ref{item:exists}) $\Rightarrow$ (\ref{item:r}): Let
\begin{math}
    u_1 \dots u_m = v_1 \dots v_n,
\end{math}
where $u_i, v_j \in \{u,v\}$. Now
\begin{equation*}
    0 = \px{u_1 \dots u_m} - \px{v_1 \dots v_n}
    = \frac{\px{u}}{X^{|u|} - 1} p - \frac{\px{v}}{X^{|v|} - 1} p
\end{equation*}
for some polynomial $p$. If $m \ne n$ or $u_i \ne v_i$ for some $i$,
then $p \ne 0$, and thus $\px{u} / (X^{|u|} - 1) = \px{v} / (X^{|v|}
- 1)$.

(\ref{item:r}) $\Rightarrow$ (\ref{item:rho}): We have
\begin{math}
    \px{u^{|v|}}
    = \px{v^{|u|}} ,
\end{math}
so $u^{|v|} = v^{|u|}$ and
\begin{math}
    \rho(u) = \rho(u^{|v|}) = \rho(v^{|u|}) = \rho(v) .
\end{math}
\end{proof}

Similarly, polynomials can be used to give a simple proof for the
theorem of Fine and Wilf. In fact, one of the original proofs in
\cite{FiWi65} uses power series. Algebraic techniques have also been
used to prove variations of this theorem \cite{MiShWa01}.

\begin{theorem}[Fine and Wilf] \label{thm:finewilf}
If $u^i$ and $v^j$ have a common prefix of length
\begin{math}
    |u| + |v| - \gcd(|u|, |v|),
\end{math}
then $\rho(u) = \rho(v)$.
\end{theorem}

\section{Solutions of Fixed Length} \label{sect:fixedlength}

In this section we apply polynomial techniques to word equations.
From now on, we will assume that the unknowns are ordered as $x_1,
\dots, x_n$ and that $\Xi$ is the set of these unknowns.

A (coefficient-free) \emph{word equation} $u = v$ on $n$ unknowns
consists of two words $u, v \in \Xi^*$. A \emph{solution} of this
equation is any morphism $h: \Xi^* \to \Sigma^*$ such that $h(u) =
h(v)$. The equation is \emph{trivial}, if $u$ and $v$ are the same
word.

The (combinatorial) \emph{rank} of a morphism $h$ is the smallest
number $r$ for which there is a set $A$ of $r$ words such that $h(x)
\in A^*$ for every unknown $x$. A morphism of rank at most one is
\emph{periodic}.

Let $h: \Xi^* \to \Sigma^*$ be a morphism. The \emph{length type} of
$h$ is the vector
\begin{equation*}
    L = (|h(x_1)|, \dots, |h(x_n)|) \in \No^n.
\end{equation*}
This length type $L$ determines a morphism $\len{L}: \Xi^* \to \No,
\len{L}(w) = |h(w)|$.

For a word equation $E: y_1 \dots y_k = z_1 \dots z_l$, where $y_i,
z_i \in \Xi$, a variable $x \in \Xi$ and a length type $L$, let
\begin{equation*}
    \qx{E}{x}{L} = \sum_{y_i = x} X^{\len{L}(y_1 \dots y_{i-1})}
        - \sum_{z_i = x} X^{\len{L}(z_1 \dots z_{i-1})} .
\end{equation*}

\begin{theorem} \label{thm:weqpeq}
A morphism $h: \Xi^* \to \Sigma^*$ of length type $L$ is a solution
of an equation $E: u = v$ if and only if
\begin{equation*}
    \sum_{x \in \Xi} \qx{E}{x}{L} \px{h(x)} = 0.
\end{equation*}
\end{theorem}
\begin{proof}
Now $h(u) = h(v)$ if and only if $\px{h(u)} = \px{h(v)}$, and the
polynomial $\px{h(u)} - \px{h(v)}$ can be written as $\sum_{x \in
\Xi} \qx{E}{x}{L} \px{h(x)}$ by \eqref{eq:prodp}.
\end{proof}

\begin{example}
Let $\Xi = \{x,y,z\}$, $E: xyz = zxy$ and $L = (1, 1, 2)$. Now
\begin{equation*}
    \qx{E}{x}{L} = 1 - X^2, \qquad
    \qx{E}{y}{L} = X - X^3, \qquad
    \qx{E}{z}{L} = X^2 - 1.
\end{equation*}
If $h$ is the morphism defined by $h(x) = 1$, $h(y) = 2$ and $h(z) =
12$, then $h$ is a solution of $E$ and
\begin{equation*}
    \qx{E}{x}{L} \px{h(x)} + \qx{E}{y}{L} \px{h(y)}
        + \qx{E}{z}{L} \px{h(z)}
    = (1 - X^2) \cdot 1 + (X - X^3) \cdot 2 + (X^2 - 1) (1 + 2X)
    = 0.
\end{equation*}
\end{example}

A morphism $\phi: \Xi^* \to \Xi^*$ is an \emph{elementary
transformation}, if there are $x, y \in \Xi$ so that $\phi(y) \in
\{xy, x\}$ and $\phi(z) = z$ for $z \in \Xi \smallsetminus \{y\}$.
If $\phi(y) = xy$, then $\phi$ is \emph{regular}, and if $\phi(y) =
x$, then $\phi$ is \emph{singular}. The next lemma follows
immediately from results in \cite{Lo83}.

\begin{lemma} \label{lem:elemtrans}
Every solution $h$ of an equation $E$ has a factorization
\begin{math}
    h = \theta \circ \phi \circ \alpha,
\end{math}
where $\alpha(x) \in \{x, \eps\}$ for all $x \in \Xi$,
\begin{math}
    \phi = \phi_m \circ \dots \circ \phi_1,
\end{math}
every $\phi_i$ is an elementary transformation and $\phi \circ
\alpha$ is a solution of $E$. If $\alpha(x) = \eps$ for $s$ unknowns
$x$ and $t$ of the $\phi_i$ are singular, then the rank of $\phi
\circ \alpha$ is $n-s-t$.
\end{lemma}

\begin{lemma} \label{lem:rdim}
Let $E: u = v$ be an equation on $n$ unknowns. Let $h: \Xi^* \to
\Sigma^*$ be a solution of length type $L$ that has rank $r$. There
is an $r$-dimensional subspace $V$ of $\Q^n$ such that $L \in V$ but
those length types of the solutions of $E$ of rank $r$ that are in
$V$ are not covered by any finite union of $(r-1)$-dimensional
spaces.
\end{lemma}
\begin{proof}
Let
\begin{math}
    h = \theta \circ \phi_m \circ \dots \circ \phi_1 \circ \alpha
\end{math}
as in Lemma \ref{lem:elemtrans}. Let $f_k = \phi_k \circ \dots \circ
\phi_1 \circ \alpha$. Now $g \circ f_m$ is a solution of $E$ for
every morphism $g: \Xi^* \to \Sigma^*$. The length type of $g \circ
f_m$ is
\begin{equation} \label{eq:2dimsol}
    \sum_{i=1}^n |g(x_i)| \cdot (|f_m(x_1)|_{x_i}, \dots, |f_m(x_n)|_{x_i})
\end{equation}
To prove the theorem, we need to show that at least $r$ of the
vectors in this sum are linearly independent.

Let $A_k$ be the $n \times n$ matrix
\begin{math}
    (|f_k(x_i)|_{x_j}).
\end{math}
If there are $s$ unknowns $x$ such that $\alpha(x) = \eps$, then the
rank of $A_0$ is $n-s$. If $\phi_k$ is regular, then the matrix
$A_{k}$ is obtained from $A_{k-1}$ by adding one of the columns to
another column, so the ranks of these matrices are equal. If
$\phi_k$ is singular, then $A_{k}$ is obtained from $A_{k-1}$ by
adding one of the columns to another column and setting some column
to zero, so the rank of the matrix is decreased by at most one. If
$t$ of the $\phi_i$ are singular, then the rank of $A_m$ is at least
$n-s-t$. The rank of $f_m$ is $n-s-t$, so $r \leq n-s-t$ and at
least $r$ of the columns of $A_m$ are linearly independent.
\end{proof}

\begin{lemma} \label{lem:rdim2}
Let $E: u = v$ be an equation and $h: \Xi^* \to \Sigma^*$ be a
solution of length type $L$ that has rank $r$. There are morphisms
$f_m: \Xi^* \to \Xi^*$ and $g_m: \Xi^* \to \Sigma^*$ and polynomials
$p_{ij}$ such that the following conditions hold:
\begin{enumerate}
\item $h = g_m \circ f_m$,
\item $f_m$ is a solution of $E$,
\item $\px{(g \circ f_m)(x_i)} = \sum p_{ij} \px{g(x_j)}$ for all
    $i,j$, if $g: \Xi^* \to \Sigma^*$ is a morphism of the same
    length type as $g_m$,
\item $r$ of the vectors $(p_{1j}, \dots, p_{nj})$, where
    $j = 1, \dots, n$, are linearly independent.
\end{enumerate}
\end{lemma}
\begin{proof}
Let $f_k$ be as in the proof of Lemma \ref{lem:rdim} and let $g_k$
be such that $h = g_k \circ f_k$. For every $k$, there are
polynomials $p_{ijk}$ so that
\begin{math}
    \px{h(x_i)} = \sum_{j=1}^n p_{ijk} \px{g_k(x_j)}
\end{math}
for all $i \in \{1, \dots, n\}$ ($p_{ijk}$ ``encodes'' the positions
of the word $g_k(x_j)$ in $h(x_i)$). Let $B_k$ be the $n \times n$
matrix
\begin{math}
    (p_{ijk}).
\end{math}
The matrix $B_{k+1}$ is obtained from $B_k$ by adding one of the
columns to another column, and multiplying some column with a
polynomial. Like in Lemma \ref{lem:rdim}, we conclude that at least
$n-s-t$ of the columns of $B_m$ are linearly independent and $r \leq
n-s-t$. If we let $p_{ij} = p_{ijm}$, then the four conditions hold.
\end{proof}

With the help of these lemmas, we are going to analyze solutions of
some fixed length type. Fundamental solutions (which were implicitly
present in the previous lemmas, see \cite{Lo83}) have been used in
connection with fixed lengths also in \cite{Ho01} and \cite{Ho00}.


\begin{theorem} \label{thm:rank}
Let $E_1, \dots, E_{m}$ be a system of equations on $n$ unknowns and
let $L \in \No^n$. Let
\begin{math}
    q_{ij} = \qx{E_i}{x_j}{L}.
\end{math}
If the system has a solution of length type $L$ that has rank $r$,
then the rank of the $m \times n$ matrix $(q_{ij})$ is at most
$n-r$. If the rank of the matrix is 1, at most one component of $L$
is zero and the equations are nontrivial, then they have the same
solutions of length type $L$.
\end{theorem}
\begin{proof}
Let $h$ be a solution of length type $L$ that has rank $r$. If
$r=1$, the first claim follows from Theorem \ref{thm:weqpeq}, so
assume that $r > 1$. Let $E$ be an equation that has the same
nonperiodic solutions as the system. We will use Lemma
\ref{lem:rdim2} for this equation. Fix $k$ and let $g: \Xi^* \to
\Sigma^*$ be the morphism determined by
\begin{math}
    g(x_k) = 10^{|g_m(x_k)|-1}
\end{math}
and
\begin{math}
    g(x_i) = 0^{|g_m(x_i)|}
\end{math}
for all $i \ne k$ (we assumed earlier that $0 \notin \Sigma$, but it
does not matter here). Then $g \circ f_m$ is a solution of every
$E_l$,
\begin{math}
    \px{(g \circ f_m)(x_i)} = \sum_{j=1}^n p_{ij} \px{g(x_j)}
\end{math}
and
\begin{equation*}
    0 = \sum_{i=1}^n \qx{E_l}{x_i}{L} \sum_{j=1}^n p_{ij} \px{g(x_j)}
    = \sum_{i=1}^n \qx{E_l}{x_i}{L} p_{ik}
\end{equation*}
for all $l$ by Theorem \ref{thm:weqpeq}. Thus the vectors $(p_{1j},
\dots, p_{nj})$ are solutions of the linear system of equations
determined by the matrix $(q_{ij})$. Because at least $r$ of these
vectors are linearly independent, the rank of the matrix is at most
$n-r$.

If at most one component of $L$ is zero and the equations are
nontrivial, then all rows of the matrix are nonzero. If also the
rank of the matrix is 1, then all rows are multiples of each other
and the second claim follows by Theorem \ref{thm:weqpeq}.
\end{proof}

\section{Applications} \label{sect:appl}

The \emph{graph} of a system of word equations is the graph, where
$\Xi$ is the set of vertices and there is an edge between $x$ and
$y$, if one of the equations in the system is of the form $x \dd = y
\dd$. The following well-known theorem can be proved with the help
of Theorem \ref{thm:rank}.

\begin{theorem}[Graph Lemma] \label{thm:graph}
Consider a system of equations whose graph has $r$ connected
components. If $h$ is a solution of this system and $h(x_i) \ne
\eps$ for all $i$, then $h$ has rank at most $r$.
\end{theorem}
\begin{proof}
We can assume that the connected components are
\begin{equation*}
    \{x_1, \dots, x_{i_2-1}\},
    \{x_{i_2}, \dots, x_{i_3-1}\},
    \dots,
    \{x_{i_r}, \dots, x_n\}
\end{equation*}
and the equations are
\begin{equation*}
    x_j \dd = x_{k_j} \dd,
\end{equation*}
where $j \in \{1, \dots, n\} \smallsetminus \{1, i_2, \dots, i_r\}$
and $k_j < j$. Let $q_{ij}$ be as in Theorem \ref{thm:rank}. If we
remove the columns $1, i_2, \dots, i_r$ from the $(n-r) \times n$
matrix $(q_{ij})$, we obtain a square matrix $M$, where the diagonal
elements are not divisible by $X$, but all elements above the
diagonal are divisible by $X$. This means that $\det(M)$ is not
divisible by $X$, so $\det(M) \ne 0$. Thus the rank of the matrix
$(q_{ij})$ is $n-r$ and $h$ has rank at most $r$ by Theorem
\ref{thm:rank}.
\end{proof}

The next theorem generalizes a result from \cite{CzKa07} for more
than three unknowns.

\begin{theorem}
If a pair of nontrivial equations on $n$ unknowns has a solution $h$
of rank $n-1$, where no two of the unknowns commute, then there is a
number $k \geq 1$ such that the equations are of the form
\begin{math}
    x_1 \dd = x_2^k x_3 \dd .
\end{math}
\end{theorem}
\begin{proof}
By Theorem \ref{thm:graph}, the equations must be of the form $x_1
\dd = x_2 \dd$. Let them be
\begin{equation*}
    x_1 u y \dd = x_2 v z \dd
    \qquad \text{and} \qquad
    x_1 u' y' \dd = x_2 v' z' \dd,
\end{equation*}
where $u, v, u', v' \in \{x_1, x_2\}^*$ and $y, z, y', z' \in \{x_3,
\dots, x_n\}$. We can assume that $z = x_3$ and
\begin{math}
    |h(x_2 v)| \leq |h(x_1 u)|, |h(x_1 u')|, |h(x_2 v')|.
\end{math}
If it would be $|h(x_1 u)| = |h(x_2 v)|$, then $h(x_1)$ and $h(x_2)$
would commute, so $|h(x_1 u)| > |h(x_2 v)|$. If $v$ would contain
$x_1$, then $h(x_1)$ and $h(x_2)$ would commute by Theorem
\ref{thm:finewilf}, so $v = x_2^{k-1}$ for some $k \geq 1$.

Let $L$ be the length type of $h$ and let $q_{ij}$ be as in Theorem
\ref{thm:rank}. By Theorem \ref{thm:rank}, the rank of the matrix
$(q_{ij})$ must be 1 and thus
\begin{math}
    q_{12} q_{23} - q_{13} q_{22} = 0.
\end{math}
The term of
\begin{math}
    q_{13} q_{22}
\end{math}
of the lowest degree is $X^{|h(x_2^k)|}$. The same must hold for
\begin{math}
    q_{12} q_{23},
\end{math}
and thus the term of $q_{23}$ of the lowest degree must be
$-X^{|h(x_2^k)|}$. This means that $|h(x_2 v')| = |h(x_2^k)| \leq
|h(x_1 u')|$ and  $z' = x_3$. As above, we conclude that $|h(x_2
v')| < |h(x_1 u')|$, $v'$ cannot contain $x_1$ and $v' = x_2^{k-1}$.
\end{proof}

It was proved in \cite{Ko98} that if
\begin{equation*}
      s_0 u_1^i s_1 \dots u_m^i s_m
    = t_0 v_1^i t_1 \dots v_n^i t_n
\end{equation*}
holds for $m+n+3$ consecutive values of $i$, then it holds for all
$i$. By using similar ideas as in Theorem \ref{thm:rank}, we improve
this bound to $m+n$ and prove that the values do not need to be
consecutive. In \cite{Ko98} it was also stated that the
arithmetization and matrix techniques in \cite{Tu87} would give a
simpler proof of a weaker result. Similar questions have been
studied in \cite{HoKo07} and there are relations to independent
systems \cite{Pl03}.

\begin{theorem}
Let $m,n \geq 1$, $s_j, t_j \in \Sigma^*$ and $u_j, v_j \in
\Sigma^+$. Let
\begin{math}
    U_i = s_0 u_1^i s_1 \dots u_m^i s_m
\end{math}
and
\begin{math}
    V_i = t_0 v_1^i t_1 \dots v_n^i t_n .
\end{math}
If $U_i = V_i$ holds for $m+n$ values of $i$, then it holds for all
$i$.
\end{theorem}
\begin{proof}
The equation $U_i = V_i$ is equivalent with $\px{U_i} - \px{V_i} =
0$.
This equation can be written as
\begin{equation} \label{eq:1}
    \sum_{j=0}^{m} y_j X^{i |u_1 \dots u_j|}
    + \sum_{k \in K} z_k X^{i |v_1 \dots v_k|} = 0,
\end{equation}
where $y_j, z_k$ are some polynomials, which do not depend on $i$,
and $K$ is the set of those $k \in \{0, \dots n\}$ for which $|v_1
\dots v_k|$ is not any of the numbers $|u_1 \dots u_j|$ ($j = 0,
\dots, m$). If $U_{i_1} = V_{i_1}$ and $U_{i_2} = V_{i_2}$, then
\begin{equation*}
    (i_1 - i_2) |u_1 \dots u_m| = |U_{i_1}| - |U_{i_2}|
    = |V_{i_1}| - |V_{i_2}| = (i_1 - i_2) |v_1 \dots v_n| .
\end{equation*}
Thus $|u_1 \dots u_m| = |v_1 \dots v_n|$ and the size of $K$ is at
most $n-1$. If \eqref{eq:1} holds for $m + 1 + \# K \leq m+n$ values
of $i$, it can be viewed as a system of equations, where $y_j, z_k$
are unknowns. The coefficients of this system form a generalized
Vandermonde matrix, whose determinant is nonzero, so the system has
a unique solution $y_j = z_k = 0$ for all $j, k$, \eqref{eq:1} holds
for all $i$ and $U_i = V_i$ for all $i$.
\end{proof}

\section{Sets of Solutions} \label{sect:solsets}

Now we analyze how the polynomials
\begin{math}
    \qx{E}{x}{L}
\end{math}
behave when $L$ is not fixed. Let
\begin{equation*}
    \lhp = \set{a_1 X_1 + \dots + a_n X_n}{a_1, \dots, a_n \in \No}
    \subset \Z[X_1, \dots, X_n]
\end{equation*}
be the additive monoid of linear homogeneous polynomials with
nonnegative integer coefficients on the variables $X_1, \dots, X_n$.
The \emph{monoid ring} of $\lhp$ over $\Z$ is the ring formed by
expressions of the form
\begin{equation*}
    a_1 X^{p_1} + \dots + a_k X^{p_k},
\end{equation*}
where $a_i \in \Z$ and $p_i \in \lhp$, and the addition and
multiplication of these generalized polynomials is defined in a
natural way. This ring is denoted by $\Z[X;\lhp]$. If $L \in \Z^n$,
then the value of a polynomial $p \in \lhp$ at the point $(X_1,
\dots X_n) = L$ is denoted by $p(L)$, and the polynomial we get by
making this substitution in $s \in \Z[X;\lhp]$ is denoted by $s(L)$.

The ring $\Z[X;\lhp]$ is isomorphic to the ring $\Z[Y_1, \dots,
Y_n]$ of polynomials on $n$ variables. The isomorphism is given by
$X^{X_i} \mapsto Y_i$. However, the generalized polynomials, where
the exponents are in $\lhp$, are suitable for our purposes.

If $a_i \leq b_i$ for $i = 1, \dots, n$, then we use the notation
\begin{equation*}
    a_1 X_1 + \dots + a_n X_n \preceq b_1 X_1 + \dots + b_n X_n.
\end{equation*}
If $p, q \in \lhp$ and $p \preceq q$, then $p(L) \leq q(L)$ for all
$L \in \No^n$.

For an equation $E: x_{i_1} \dots x_{i_r} = x_{j_1} \dots x_{j_s}$
we define
\begin{equation*}
    \sx{E}{x} = \sum_{x_{i_k} = x} X^{X_{i_1} + \dots + X_{i_{k-1}}}
        - \sum_{x_{j_k} = x} X^{X_{j_1} + \dots + X_{j_{k-1}}}
    \in \Z[X;\lhp].
\end{equation*}
Now $\sx{E}{x}(L) = \qx{E}{x}{L}$. Theorem \ref{thm:weqpeq} can be
formulated in terms of the generalized polynomials $\sx{E}{x}$.

\begin{theorem}
A morphism $h: \Xi^* \to \Sigma^*$ of length type $L$ is a solution
of an equation $E$ if and only if
\begin{equation*}
    \sum_{x \in \Xi} \sx{E}{x}(L) \px{h(x)} = 0.
\end{equation*}
\end{theorem}

\begin{example}
Let $E: x_1 x_2 x_3 = x_3 x_1 x_2$. Now
\begin{equation*}
    \sx{E}{x_1} = 1 - X^{X_3}, \qquad
    \sx{E}{x_2} = X^{X_1} - X^{X_1 + X_3}, \qquad
    \sx{E}{x_3} = X^{X_1 + X_2} - 1.
\end{equation*}
\end{example}

The \emph{length} of an equation $E: u = v$ is $|E| = |uv|$.


\begin{theorem} \label{thm:cover}
Let $E_1, E_2$ be a pair of nontrivial equations on $n$ unknowns
that don't have the same sets of solutions of rank $n-1$. The length
types of solutions of the pair of rank $n-1$ are covered by a union
of $|E_1|^2$ $(n-1)$-dimensional subspaces of $\Q^n$. If $V_1,
\dots, V_m$ is a minimal such cover and $L \in V_i$ for some $i$,
then $E_1$ and $E_2$ have the same solutions of length type $L$ and
rank $n-1$.
\end{theorem}
\begin{proof}
Let
\begin{math}
    s_{ij} = \sx{E_i}{x_j}
\end{math}
for $i = 1, 2$ and $j = 1, \dots, n$. If all $2 \times 2$ minors of
the $2 \times n$ matrix $(s_{ij})$ are zero, then for all length
types $L$ of solutions of rank $n-1$ the rank of the matrix
$(q_{ij})$ in Theorem \ref{thm:rank} is 1 and $E_1$ and $E_2$ are
equivalent, which is a contradiction. Thus there are $k,l$ such that
$t_{kl} = s_{1k} s_{2l} - s_{1l} s_{2k} \ne 0$. The generalized
polynomial $t_{kl}$ can be written as
\begin{equation*}
   t_{kl} = \sum_{i=1}^M X^{p_i} - \sum_{i=1}^N X^{q_i},
\end{equation*}
where $p_i, q_i \in \lhp$ and $p_i \ne q_j$ for all $i,j$. If $L$ is
a length type of a solution of rank $n-1$, then $M=N$ and $L$ must
be a solution of the system of equations
\begin{equation} \label{eq:ssyst}
    p_i = q_{\sigma(i)} \qquad (i=1,\dots,M)
\end{equation}
for some permutation $\sigma$. For every $\sigma$ the equations
determine an at most $(n-1)$-dimensional space.

Let
\begin{equation*}
    s_{1k} = \sum_i X^{a_i} - \sum_i X^{a'_i},
    \quad
    s_{2l} = \sum_i X^{b_i} - \sum_i X^{b'_i},
    \quad
    s_{1l} = \sum_i X^{c_i} - \sum_i X^{c'_i},
    \quad
    s_{2k} = \sum_i X^{d_i} - \sum_i X^{d'_i},
\end{equation*}
where $a_i \preceq a_{i+1}$, $a'_i \preceq a'_{i+1}$, and so on. The
polynomials $p_i$ form a subset of the polynomials $a_i + b_j$,
$a'_i + b'_j$, $c_i + d'_j$ and $c'_i + d_j$ (the reason that they
form just a subset is that we assumed $p_i \ne q_j$ for all $i,j$).
For any $i$, let $j_i$ be the smallest index $j$ such that $a_i +
b_j = p_m$ for some $m$. Now for every $i,j,m$ such that $a_i + b_j
= p_m$ we have $a_i + b_{j_i} \preceq p_m$. We can do a similar
thing for the polynomials $a'_i, b'_i$ and $c_i, d'_i$ and $c'_i,
d_i$. In this way we obtain at most $|E_1|$ polynomials $p_i$ such
that for any $L$ the value of one of these polynomials is minimal
among the values $p_i(L)$. Similarly we obtain at most $|E_1|$
``minimal'' polynomials $q_i$. It is sufficient to consider only
those systems \eqref{eq:ssyst}, where one of the equations is formed
by these ``minimal'' polynomials $p_i, q_i$. There are at most
$|E_1|^2$ possible pairs of such polynomials, and each of them
determines an $(n-1)$-dimensional space.

Consider the second claim. Because the cover is minimal, there is a
solution of rank $n-1$ whose length type is in $V_i$, but not in any
other $V_j$. By Lemma \ref{lem:rdim}, the length types of solutions
of rank $n-1$ in this space cannot be covered by a finite union of
$(n-2)$-dimensional spaces. Thus one of the systems \eqref{eq:ssyst}
must determine the space $V_i$. The same holds for systems coming
from all other nonzero $2 \times 2$ minors of the matrix $(s_{ij})$,
so $E_1$ and $E_2$ have the same solutions of rank $n-1$ and length
type $L$ for all $L \in V_i$ by Theorem \ref{thm:rank}.
\end{proof}

The following example illustrates the proof of Theorem
\ref{thm:cover}. It gives a pair of equations on three unknowns,
where the required number of subspaces is two. We do not know any
example, where more spaces would be necessary.

\begin{example}
Consider the equations
\begin{math}
    E_1: x_1 x_2 x_3 = x_3 x_1 x_2
\end{math}
and
\begin{math}
    E_2: x_1 x_2 x_1 x_3 x_2 x_3 = x_3 x_1 x_3 x_2 x_1 x_2
\end{math}
and the generalized polynomial
\begin{equation*}
\begin{split}
    s =& \sx{E_1}{x_1} \sx{E_2}{x_3}
        - \sx{E_1}{x_3} \sx{E_2}{x_1} \\
    =& X^{2 X_1 + X_2} + X^{2 X_1 + 2 X_2 + X_3} + X^{X_1 + 2 X_3}
        + X^{X_1 + X_2 + X_3}
    - X^{2 X_1 + X_2 + X_3} - X^{X_1 + X_3}
        - X^{2 X_1 + 2 X_2} - X^{X_1 + X_2 + 2 X_3}.
\end{split}
\end{equation*}
If $L$ is a length type of a nontrivial solution of the pair $E_1,
E_2$, then $s(L) = 0$. If $s(L) = 0$, then $L$ must satisfy an
equation $p = q$, where
\begin{math}
    p \in \{2 X_1 + X_2, X_1 + 2 X_3, X_1 + X_2 + X_3\}
\end{math}
and
\begin{math}
    q \in \{X_1 + X_3, 2 X_1 + 2 X_2\}.
\end{math}
The possible relations are
\begin{equation*}
    X_3 = 0, \qquad
    X_1 + X_2 = X_3, \qquad
    X_2 = 0, \qquad
    X_1 + 2 X_2 = 2 X_3.
\end{equation*}
If $L$ satisfies one of the first three, then $s(L) = 0$. If $L$
satisfies the last one, then $s(L) \ne 0$, except if $L = 0$. So if
$h$ is a nonperiodic solution, then
\begin{equation*}
    |h(x_3)| = 0 \quad \text{or} \quad
    |h(x_1 x_2)| = |h(x_3)| \quad \text{or} \quad
    |h(x_2)| = 0.
\end{equation*}
There are no nonperiodic solutions with $h(x_2) = \eps$, but every
$h$ with $h(x_3) = \eps$ or $h(x_1 x_2) = h(x_3)$ is a solution.
\end{example}

An equation $u = v$ is \emph{balanced}, if $|u|_x = |v|_x$ for every
unknown $x$. In \cite{HaNo03} it was proved that if an independent
pair of equations on three unknowns has a nonperiodic solution, then
the equations must be balanced. With the help of Theorem
\ref{thm:cover} we get a significantly simpler proof and a
generalization for this result.

\begin{theorem} \label{thm:balance}
Let $E_1, E_2$ be a pair of equations on $n$ unknowns having a
solution of rank $n-1$. If $E_1$ is not balanced, then every
solution of $E_1$ of rank $n-1$ is a solution of $E_2$.
\end{theorem}
\begin{proof}
The length types of solutions of $E_1$ are covered by a single
$(n-1)$-dimensional space $V$. Because the pair $E_1, E_2$ has a
solution of rank $n-1$, $V$ is a minimal cover for the length types
of the solutions of the pair of rank $n-1$. By Theorem
\ref{thm:cover}, $E_1$ and $E_2$ have the same solutions of length
type $L$ and rank $n-1$ for all $L \in V$.
\end{proof}

Another way to think of this result is that if $E_1$ is not balanced
but has a solution of rank $n-1$ that is not a solution of $E_2$,
then the pair $E_1, E_2$ causes a larger than minimal defect effect.

%

\section{Independent Systems} \label{sect:indsyst}

A system of word equations $E_1, \dots, E_m$ is \emph{independent},
if for every $i$ there is a morphism that is not a solution of
$E_i$, but is a solution of all the other equations.

A sequence of equations $E_1, \dots, E_m$ is a \emph{chain}, if for
every $i$ there is a morphism that is not a solution of $E_i$, but
is a solution of all the preceding equations.

The question of the maximal size of an independent system is open.
Only things that are known are that independent systems cannot be
infinite and there are systems of size $\Theta(n^4)$, where $n$ is
the number of unknowns. For a survey on these topics, see
\cite{KaSa11}.

We study the following variation of the above question: how long can
a sequence of equations $E_1, \dots, E_m$ be, if for every $i$ there
is a morphism of rank $n-1$ that is not a solution of $E_i$, but is
a solution of all the preceding equation? We prove an upper bound
depending quadratically on the length of the first equation. For
three unknowns we get a similar bound for the size of independent
systems and chains.

\begin{theorem} \label{thm:chain}
Let $E_1, \dots, E_m$ be nontrivial equations on $n$ unknowns having
a common solution of rank $n-1$. For every $i \in \{1, \dots,
m-1\}$, assume that there is a solution of the system $E_1, \dots,
E_i$ of rank $n-1$ that is not a solution of $E_{i+1}$. If the
length types of solutions of the pair $E_1, E_2$ of rank $n-1$ are
covered by a union of $N$ $(n-1)$-dimensional subspaces, then $m
\leq N + 1$. In general, $m \leq |E_1|^2 + 1$.
\end{theorem}
\begin{proof}
We can assume that $E_i$ is equivalent with the system $E_1, \dots,
E_i$ for all $i \in \{1, \dots, m\}$. Let the length types of
solutions of $E_2$ of rank $n-1$ be covered by the
$(n-1)$-dimensional spaces $V_1, \dots, V_N$. Some subset of these
spaces forms a minimal cover for the length types of solutions of
$E_3$ of rank $n-1$. If this minimal cover would be the whole set,
then $E_2$ and $E_3$ would have the same solutions of rank $n-1$ by
the second part of Theorem \ref{thm:cover}. Thus the length types of
solutions of $E_3$ of rank $n-1$ are covered by some $N-1$ of these
spaces. We conclude inductively that the length types of solutions
of $E_i$ of rank $n-1$ are covered by some $N-i+2$ of these spaces
for all $i \in \{2, \dots, m\}$. It must be $N-m+2 \geq 1$, so $m
\leq N + 1$. By the first part of Theorem \ref{thm:cover}, $N \leq
|E_1|^2$.
\end{proof}

In Theorem \ref{thm:chain} it is not enough to assume that the
equations are independent and have a common solution of rank $n-1$.
If the number of unknowns is not fixed, then there are arbitrarily
large such systems, where the length of every equation is 10, see
e.g. \cite{HaKaPl02}.

In the case of three unknowns, Theorem \ref{thm:chain} gives an
upper bound depending on the length of the shortest equation for the
size of an independent system of equations, or an upper bound
depending on the length of the first equation for the size of a
chain of equations. A better bound in Theorem \ref{thm:cover} would
immediately give a better bound in the following corollary.

\begin{corollary}
If $E_1, \dots, E_m$ is an independent system on three unknowns
having a nonperiodic solution, then $m \leq |E_1|^2 + 1$. If $E_1,
\dots, E_m$ is a chain of equations on three unknowns, then $m \leq
|E_1|^2 + 5$.
\end{corollary}

\bibliographystyle{eptcs}
\bibliography{ref_poly}

\end{document}